\documentclass[conference]{IEEEtran}
\IEEEoverridecommandlockouts
\usepackage{cite}
\usepackage{amsmath,amssymb,amsfonts}
\usepackage{verbatim}
\usepackage{algorithmic}
\usepackage{graphicx}
\usepackage{textcomp}
\usepackage{bm}
\usepackage{xcolor}
\usepackage{amsthm}
\usepackage{float}
\usepackage{subfigure}
\usepackage{epstopdf}
\newtheorem{Le}{Lemma}
\newtheorem{theo}{Theorem} 
\theoremstyle{definition}
\newtheorem{re}{Remark}
\newcommand{\PP}{\mathcal{P}}
\newcommand{\XX}{\mathcal{X}}
\newcommand{\YY}{\mathcal{Y}}
\newcommand{\DD}{\mathcal{D}}
\newcommand{\CC}{\mathcal{C}}
\newcommand{\BB}{\mathcal{B}}
\newcommand{\MS}{\mathcal{S}}
\newcommand{\xx}{\bm{\mathrm{x}}}
\newcommand{\yy}{\bm{\mathrm{y}}}
\newcommand{\zz}{\bm{\mathrm{z}}}
\newcommand{\cc}{\bm{\mathrm{c}}}

\newcommand{\BN}{\bm{N}}
\newcommand{\CN}{\mathcal{CN}}
\newcommand{\FR}{\mathfrak{R}}
\newcommand{\FI}{\mathfrak{I}}
\newcommand{\hh}{\hat{h}}
\newcommand{\uh}{\underline{h}}
\newcommand{\oh}{\overline{h}}
\newcommand{\ola}{\lambda^{ave}}
\newcommand{\SM}{M^*(n,P,\lambda_1,\lambda_2)}
\newcommand{\er}{\epsilon(\rho)}
\def\BibTeX{{\rm B\kern-.05em{\sc i\kern-.025em b}\kern-.08em
    T\kern-.1667em\lower.7ex\hbox{E}\kern-.125emX}}
\begin{document}

\title{Deterministic Identification over Channels without CSI} 

 \author{%
   \IEEEauthorblockN{Yuan Li\IEEEauthorrefmark{2}\IEEEauthorrefmark{3}\IEEEauthorrefmark{1},
                     Xianbin Wang\IEEEauthorrefmark{1},
                     Huazi Zhang\IEEEauthorrefmark{1},                    
                     Jun Wang\IEEEauthorrefmark{1},    
                     Wen Tong\IEEEauthorrefmark{1},                 
                     Guiying Yan\IEEEauthorrefmark{2}\IEEEauthorrefmark{3}
                     and Zhiming Ma\IEEEauthorrefmark{2}\IEEEauthorrefmark{3}}
  \IEEEauthorblockA{\IEEEauthorrefmark{2}%
                     University of Chinese Academy of Sciences, \IEEEauthorrefmark{3}%
                     Academy of Mathematics and Systems Science, CAS}
   \IEEEauthorblockA{\IEEEauthorrefmark{1}%
                     Huawei Technologies Co. Ltd.}
    Email: liyuan181@mails.ucas.ac.cn, \{zhanghuazi, wangxianbin1, justin.wangjun, tongwen\}@huawei.com,\\
         yangy@amss.ac.cn, mazm@amt.ac.cn  }
\maketitle

\begin{abstract}

Identification capacities of randomized and deterministic identification were proved to exceed channel capacity for Gaussian channels \emph{with} channel side information (CSI). In this work, we extend deterministic identification to the block fading channels \emph{without} CSI by applying identification codes for both channel estimation and user identification. We prove that identification capacity is asymptotically higher than transmission capacity even in the absence of CSI. And we also analyze the finite-length performance theoretically and numerically. The simulation results verify the feasibility of the proposed blind deterministic identification in finite blocklength regime.
\end{abstract}

\begin{IEEEkeywords}
deterministic identification, fading channels
\end{IEEEkeywords}

\section{Introduction}

In the classical Shannon transmission framework \cite{b2}, a sender transmits a message from message set $\mathcal{M}$ to the  receiver through a noisy channel $W(y|x)$. And the receiver needs to correctly decode that message. In \cite{b3}, Ahlswede and Dueck introduced a different identification framework. In identification,  receiver-$j$ only cares whether message-$j$ is transmitted. If receiver-$j$ believes message-$j$ is not sent, it does not attempt to decode that message. In other words, the receiver faces a hypothesis-testing problem with two hypotheses, which is much simpler compared to the Shannon transmission problem with $|\mathcal{M}|$ hypotheses. 

In randomized identification, each receiver has its own codebook, and message-$j$ is sent by randomly transmitting a codeword from receiver-$j$'s codebook. Receivers decide if $y$ belongs to the decoding region of a codeword from its own codebook. A type-I error occurs when the target receiver misses its message, and a type-II error occurs when the non-target receiver accepts a false message. In \cite{b3}\cite{b4}, the number of messages were  proved to scale double exponentially in the blocklength with vanishing type-I and type-II error rates.

 Identification framework was generalized to broadcast channel \cite{b5}, Multiple-Input and Multiple-Output (MIMO) channel \cite{b6} and compound channel (CC) \cite{b7a}. And feedback was proved to incease the identification capacity \cite{b7} of memoryless channels. 

Explicit constructions of constant-weight codes for identification were proposed in \cite{b8}\cite{b9}. However, the implementation requires computation in a large alphabet, which suffers from high complexity and latency. In deterministic identification, each receiver has a codebook with only one codeword, thus local randomness is no longer needed, which greatly reduces the computation complexity. Deterministic identification via binary symmetric channels was first discussed in \cite{b10}. In \cite{b3}, deterministic identification capacity was provided for discrete memoryless channel (DMC) without proof, and in the later work \cite{r10}, there was a gap in the proof. Finally, the complete proof was provided in \cite{b11}. Even as the number of messages grows exponentially, deterministic identification capacity can still be significantly higher than channel capacity for DMC, and  is proved to be infinite in the exponential scale over continuous input fading channel \emph{with} CSI \cite{b12}. 

Although deterministic identification is attractive from the theoretical point of view, its benefits deserve further study in practical scenarios. In this paper, we prove some promising results under very practical constraints, i.e., deterministic identification with unknown/inaccurate CSI in finite-length regime. In particular, we coin deterministic identification \emph{without} CSI as \emph{blind deterministic identification}. Though the decoding regions of identification codes can be constructed in a maximum-likelihood (ML) manner without channel estimation, it is not considered in this paper due to high complexity. In this work, we utilize channel estimation to compute decoding metrics. On the one hand, these metrics do not reduce the achievable rate asymptotically. On the other hand,  channel estimation is necessary in most scenarios. 

In this paper, we prove that in block fading channel, the blind identification rate can still exceed transmission rate. Moreover, we analyze the type-I and type-II error rates in finite blocklength regime. As expected, there is a tradeoff between the missed detection rate (type-I error) and the false activation rate (type-II error). 

In practice, type-I errors are more harmful than type-II errors because the latter can be further reduced by a cyclic redundancy check (CRC). If a type-II error occurs in identification codes, a non-target user considers itself as the target user at first, it will know the truth after the CRC check fails. Nevertheless, if a type-I error occurs, the target user will discard the message, which cannot be recovered in subsequent decoding. In view of the reasons mentioned above, in our design of the identification code, we bound the type-I error rate below the target error rate while minimize type-II error rate. If the type-II error rate is higher than the target error rate, CRC can bound the type-II error rate below the target error rate. For example, we apply blind identification codes as pilot symbols in pilot-assisted transmission (PAT) schemes for both channel estimation and user identification. Simulation results verify the high efficiency of the blind identification performance in finite blocklength regime. 

This paper is organized as follows. In section II, we review the background of identification and deterministic identification. In section III we provide the identification capacity analysis in both asymptotic and nonasymptotic regimes. Simulation results are provided in section IV for block fading channels. Finally we draw conclusions in section V.

\section{Background}
\subsection{Identification}
A stochastic matrix $W=\{W(y|x), \ x \in \XX, \ y \in \YY)\}$ defines a DMC, where $\XX$ and $\YY$ are finite sets with the transmission probability
\begin{align}
W^n(y^n|x^n) = \prod\limits_{i=1}^n W(y_i|x_i),
\end{align}
for length-$n$ sequences $x^n = (x_1, \cdots, x_n) \in \XX^n$, $y^n = (y_1, \cdots, y_n) \in \YY^n$. 

For DMC $W=\{W(y|x), \ x \in \XX, \ y \in \YY)\}$, let $\PP(\XX^n)$ denote the set of probability distributions on $\XX^n$, a randomized identification code $(n,M,\lambda_1, \lambda_2)$ is defined by a family 
$\{(Q(\cdot|i),\DD_i)|i=1, \cdots ,M\}$,
where $Q(\cdot|i) \in \PP(\XX^n)$ are randomized encoders, $\DD_i \subseteq \YY^n$ are decoding regions. The type-I and type-II error rates satisfy $\forall \ i,j=1,\cdots,M, i \neq j$
\begin{equation}
\sum_{x^n \in \XX^n } Q(x^n|i) W^n(\DD_i^{c}|x^n) \leq \lambda_1, 
\end{equation}
\begin{equation}
\sum_{x^n \in \XX^n } Q(x^n|i) W^n(\DD_j|x^n) \leq \lambda_2.
\end{equation}

The dominating difference between identification and transmission is that decoding regions $\DD_i$ do not have to be disjoint. Therefore, the number of supported messages is much higher.   
 
$R$ is a $(\lambda_1,\lambda_2)$ achievable rate, if $\forall \ \gamma > 0$ and $n$ is sufficiently large, there exist $(n,M,\lambda_1,\lambda_2)$ identification codes that satisfy
\begin{equation}
\frac{1}{n}\log \log M \geq R - \gamma.
\end{equation}
The supremum of achievable identification rates is called identification capacity $C_{I}(W)$ of the DMC $W$.

In \cite{b3}\cite{b4}, capacity-achieving codes were used to construct randomized identification codes over DMC. If $\lambda_1 + \lambda_2 < 1$, identification capacity was proved to be equal to transmission capacity, i.e. $C_{I}(W) = C(W),$
where $C(W)$ is channel capacity \cite{b2}. Similar results were proved for compound channels in \cite{b7a}. 

A direct result from the proof is that, $Q(\cdot|i)$ can be chosen uniformly over proper subsets of capacity-achieving transmission codes. 
In \cite{b8}\cite{b9}, two layers of Reed-Solomon codes were used to explicit construct capacity-achieving identification codes. However, the latest implementation still incurs high complexity due to the computation in large finite fields \cite{b9a}.

\subsection{Deterministic Identification}
Deterministic identification codes can be constructed by choosing $Q(\cdot|i)$ to be the point mass on codeword $c_i$. The type-I and type-II error rates are given by
\begin{align}
P_1(i)& = W^n(\DD_i^c|c_i), \ P_2(i,j) = W^n(\DD_j|c_i),
\end{align}
where $\DD_i \subseteq \YY^n$ is the decoding region for message-$i$, $P_1(i)$ is the type-I error rate of message-$i$  and $P_2(i,j)$  is the probability that message-$i$ is misidentified as message-$j$.

Though the number of messages scales exponentially in the blocklength, i.e., $M\approx 2^{nC_{DI}(W)}$, the achievable identification capacity $C_{DI}(W)$ can still be significantly higher than transmission capacity $C(W)$. This is mainly due to the relaxation that allows overlapping in the decoding regions, as shown in Fig. \ref{fig.2}. 

\begin{figure}[htbp]\centering
\subfigure[Transmission]{
\label{fig.1}
\includegraphics[width = .22\textwidth]{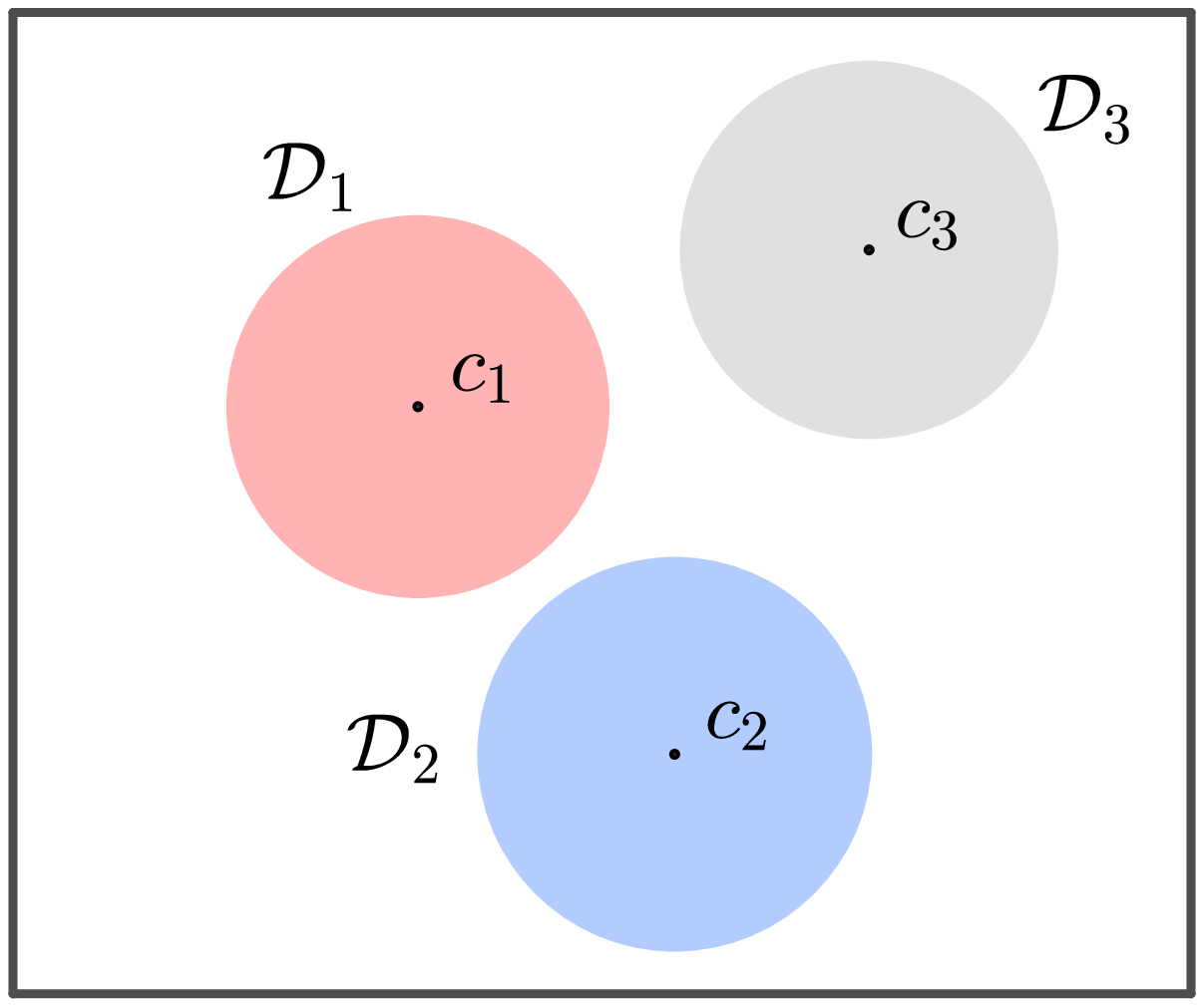}
}
\hfill
\subfigure[Identification]{
\label{fig.2}
\includegraphics[width = .22\textwidth]{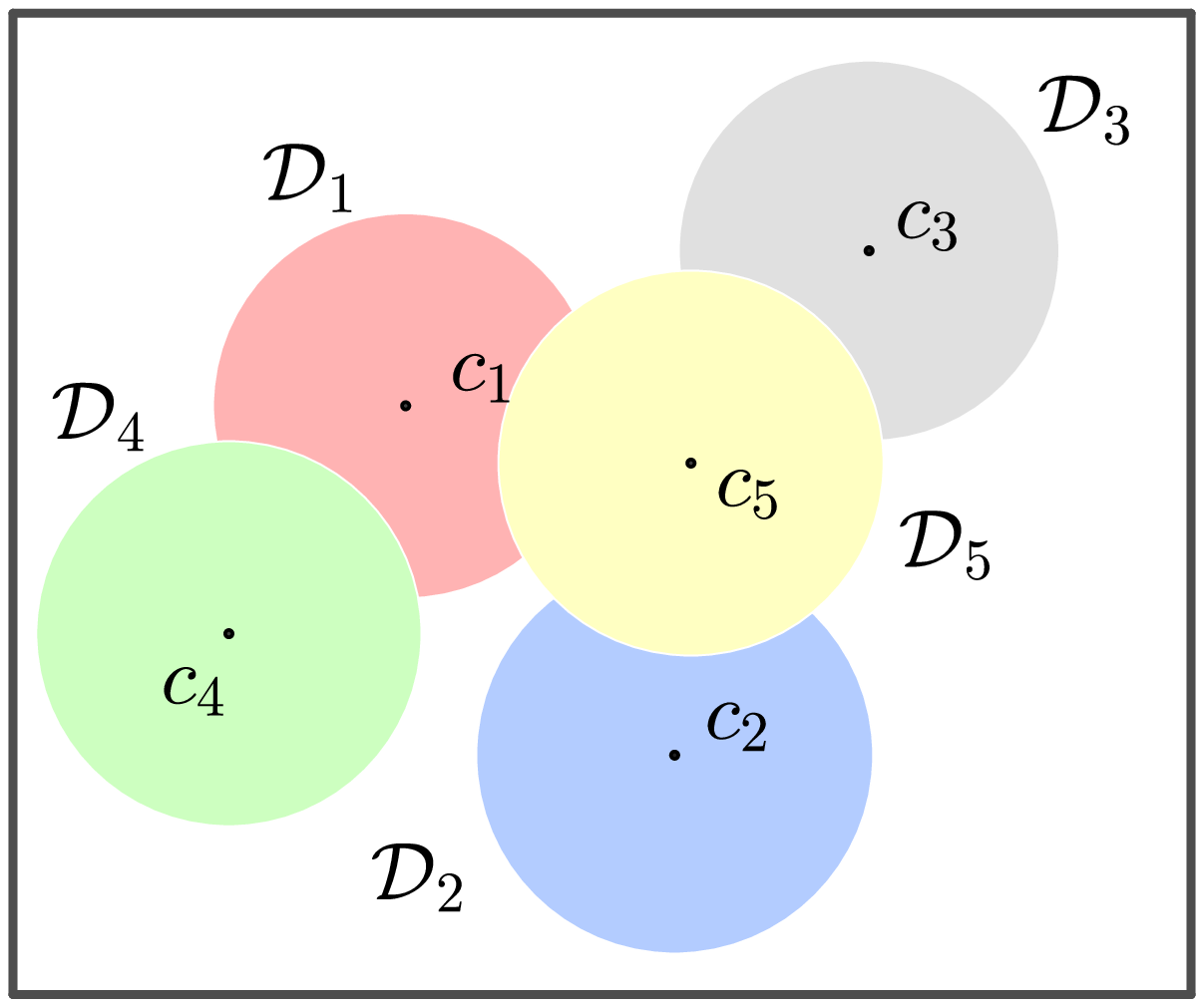}
}
\caption{No disjointness is imposed on $\DD_i$ in identification codes, thus the supported number of messages is much higher than that in transmission codes.}
\label{fig3}
\end{figure}  

\ \ \ \ In \cite{b11}, deterministic identification capacity $C_{DI}(W)$ of DMC $W$ without any constraints was proved to be
$C_{DI}(W) = \log n_{row}(W)$,
where $n_{row}$ is the number of distinct rows of $W$. For example, $C_{DI}(W)=1$ regardless of the cross probability $0 < p<\frac{1}{2}$ of binary symmetric channel, whose transmission capacity $C(W)=1-H(p)$, where $H(p) = -p\log p-(1-p)\log (1-p) > 0$.

To summarize, deterministic identification codes have simpler implementation and higher achievable identification rate than transmission codes. They may lead to higher energy and spectrum efficiencies in a wireless system.  

\subsection{Identification over Fading Channels \emph{with} CSI}
In \cite{b12}, deterministic identification was generalized to infinite input alphabet. In this scenario, deterministic identification capacity is infinite in the exponential scale regardless of channel noise. If fading coefficients are positive and bounded away from zero, the deterministic identification capacity of Gaussian channel \emph{with} CSI was proved to scale as $n \log n$ \cite{b12}, i.e., $M \approx 2^{(n\log n)R}$.    

\section{Identification over Channels \emph{without} CSI}
In this section, we extend deterministic identification to block fading Gaussian channels \emph{without} CSI. Firstly, we prove that deterministic identification capacity exceeds transmission capacity asymptotically even when CSI is unknown. Secondly, we analyze the finite-length identification performance. In an application example, we use blind identification codes in PAT schemes for both channel estimation and user identification. 
\subsection{Notations and Definitions}
In this paper, $\log x$ is the base-2 logarithm of $x$. $|\mathcal{A}|$ is the cardinality of set $\mathcal{A}$.

Let $|x| = \sqrt{\FR(x)^2+\FI(x)^2}$ be the norm of complex number $x = \FR(x)+j\FI(x)$, where $\FR(x),\FI(x)$ are the real and imaginary parts of $x$. Let $\arg x$ be the argument of $x$, and $\overline{x}=\FR(x)-j\FI(x)$ be the conjugate of $x$.

Let $\xx=(x_1,\cdots,x_n)^T$ be an $n$-dimensional complex vector, with the $l_2$-norm  
\begin{equation}
\label{eq1}
||\xx|| = \sqrt{\sum_{k=1}^n \FR(x_k)^2+ \FI(x_k)^2}.
\end{equation}
Denote by $\xx^H =(\overline{x}_1,\cdots,\overline{x}_n)$  the conjugate transpose of $\xx$. The notation $\langle \xx,\yy \rangle=\yy^H\xx$ represents the inner product of $\xx$ and $\yy$.

Denote the $n$-dimensional complex closed sphere of radius $r$ centered at $\xx_0$ by 
\begin{align} 
\BB_{\xx_0}(n,r) = \{\xx \in \mathbb{C}^n:||\xx-\xx_0|| \leq r\}.
\end{align} 
Correspondingly, 
\begin{align} 
\MS_{\xx_0}(n,r) = \{\xx \in \mathbb{C}^n:||\xx-\xx_0|| = r\}
\end{align}
denotes the surface of $\BB_{\xx_0}(n,r)$. In view of (\ref{eq1}), $\BB_{\xx_0}(n,r)$ and $\MS_{\xx_0}(n,r)$ can also be regarded as a $2n$-dimensional real sphere and its spherical surface. 
Let $S[\cdot]$ ($V[\cdot]$) be the area (volume) of subset of $\MS_{\xx_0}(n,r)$ $(\BB_{\xx_0}(n,r))$. We abbreviate $S[\MS_{\xx_0}(n,r)]$ $(V[\BB_{\xx_0}(n,r)])$ to $S(n,r)=\frac{2n\pi^nr^{2n-1}}{n!}$ ($V(n,r)=\frac{\pi^nr^{2n}}{n!}$).

Let $X_k \sim \mathcal{N}(a_k,1)$ be independent normal random variables, then $X = \sum\limits_{k=1}^n X_k^2 \sim \chi^2(n,\delta)$-the $\chi^2$-distribution with $n$ degrees of freedom and noncentrality parameter $\delta$, where $\delta = \sum\limits_{k=1}^{n}a_k^2$. When $\delta=0$, we abbreviate $\chi^2(n,0)$ to $\chi^2(n)$. Let $F^{-1}_{n,\delta}(\cdot)$ and $F^{-1}_{n}(\cdot)$ denote the inverse cumulative distribution functions of $\chi^2(n,\delta)$ and $\chi^2(n,0)$ respectively. 

The block fading Gaussian channel with power constraints is defined as 
\begin{equation}
\yy = h\xx + \BN,
\end{equation}
where $\xx$, $\yy$ are length-$n$ input-output complex symbols, $h \sim P_h$ is the fading coefficient of the whole block and is not known at both transmitter and receiver. $\BN \sim \CN(0,\sigma^2I_n)$ is a complex Gaussian noise vector and the transmission power is limited to $||\xx||^2\leq nP$. Denote $\text{SNR}=\frac{P}{\sigma^2}$, w.l.o.g., we assume $\sigma^2 = 1, \text{SNR} = P$.

A $\left(n, M, P,\lambda_1,\lambda_2 \right)$ deterministic identification code for block fading channel \emph{without} CSI is defined by a codebook $\{c_i\}_{i=1}^M$ and a collection of decoding regions $\{\DD_i\}_{i=1}^M$. Due to the absence of CSI, we can not construct decoding regions from $h$.  A $\left(n, M, P,\lambda_1,\lambda_2 \right)$ deterministic identification code should satisfy
\begin{align}
\max_i P(\DD_i^c|c_i) \leq \lambda_1, \ \max_{i\neq j} P(\DD_j|c_i) \leq \lambda_2,
\end{align}
and $||c_i||^2 \leq nP$. The average type-II error rate for receiver-$i$ is defined as
\begin{align}
\ola_2 = \frac{1}{M-1}\sum_{j\neq i}P(\DD_j|c_i),
\end{align} 
which characterizes the average false activation rate of all the non-target receivers when $i$ is the target receiver. $\lambda_2$ depends on minimum code distance, while $\ola_2$ depends on the code spectrum. For a specific non-target user, the maximal type-II error rate is important. However, from a practical point of view, the average type-II error rate is also important, because it indicates the percentage of users who are incorrectly activated at the same time.

Let $M^*(n,P,\lambda_1,\lambda_2) =$
\begin{align}
 \max \{ M:\exists \left(n, M, P,\lambda_1,\lambda_2 \right)\text{codes}\}
\end{align}
be the maximum achievable code size. For simplicity, we only consider codes with equal-power constraints, i.e., $\forall \ i, ||c_i||^2=nP$. 

We call inputs are transmitted by Quadrature phase-shift keying (QPSK) symbols, if $c_i \in \{ \sqrt{P}e^{j(\frac{\pi}{4}+\frac{\pi}{2}k)}\}^n$, $k=0,1,2,3$. Let $n_c^k$ denote the number of $e^{j(\frac{\pi}{4}+\frac{\pi}{2}k)}$ in length-$n$ QPSK sequence $c$.
\subsection{Asymptotic Analysis}
In this section, we analyze the asymptotic achievable rate of block fading channel \emph{without} prior knowledge of CSI. Since $h$ is unknown, we first estimate $h$ as $\hh$, and $\hh$ will be used to compute identification metric. The following lemma determines the distributions of $|\hh|$ and mean square error $(mse)$ for both target and non-target receivers.

\begin{Le}
\label{le1}
For receiver j, let $\hh_j = \frac{c_j^H \yy}{nP}$, $\hat{\yy}_j= \hh_j c_j$ and $ mse_j = (\yy-\hat{\yy}_j)^H(\yy-\hat{\yy}_j)$.\\
If $j$ is the target receiver, i.e., $\yy=hc_j+\BN$, then
\begin{align}
&2nP|\hh_j|^2 \sim \chi^2(2,2nP|h|^2), \label{eq71} \\
&2mse_j \sim \chi^2(2n-2), \label{eq72}
\end{align}
and if $i$ is the target receiver, i.e., $\yy =hc_i+\BN$, then
\begin{align}
&2nP|\hh_j|^2 \sim \chi^2(2,2nP|h|^2\frac{|\langle c_i, c_j\rangle|^2}{(nP)^2}), \label{eq81} \\
&2mse_j \sim \chi^2\left(2n-2,2nP|h|^2(1-\frac{|\langle c_i,c_j \rangle|^2}{(nP)^2})\right). \label{eq82}
\end{align}
Furthermore, $2nP|\hh_j|^2 $ and $mse_j$ are independent for any fixed $h$.
\end{Le}

\begin{proof}
The proof is mainly based on generalizing the distribution of quadratic forms of real multivariate normal vectors \cite{bl} to those of complex normal vectors. 

Let $\yy = hc_i+\BN$, then
\begin{align}
nP|\hh_j|^2 &= \left(\BN+hc_i\right)^H \frac{c_jc_j^H}{nP}\left(\BN+hc_i\right), \\
mse_j &= \left(\BN+hc_i\right)^H\left(I_n-\frac{c_jc_j^H}{nP}\right)\left(\BN+hc_i\right).
\end{align}
Denote $A = \frac{c_jc_j^H}{nP}$, $B = I_n-\frac{c_jc_j^H}{nP}$, we have $A = A^2=A^H$, $B = B^2=B^H$ and \text{rank}$(B)= \ $\text{trace}$(B)=n-1$.
Thus $B$ can be unitary decomposed  into $B=D\Sigma D^H$, where $
\Sigma=\left(\begin{array}{ll}
I_{n-1} & \bm{0} \\
\bm{0}^T & 0
\end{array}\right)
$, $D = \left(\zz_1,\cdots,\zz_{n-1},\frac{c_j}{\sqrt{nP}}\right)$, and $\zz_i$ are unit orthogonal eigenvectors of eigenvalue $1$.
Therefore
\begin{align}
nP|\hh_j|^2 &= \left(\FR(N'_n+c'_n)\right)^2 + \left(\FI(N'_n+c'_n)\right)^2,\\
mse_j &= \left(\BN'+\cc'\right)^H \Sigma \left(\BN'+\cc'\right)\notag \\ 
& = \sum_{k=1}^{n-1}\left(\FR(N'_k+c'_k)\right)^2 + \left(\FI(N'_k+c'_k)\right)^2,
\end{align}
where $\BN' = D^H \BN \sim \CN(0,I_n)$\cite[Theorem 1]{b13}, $\cc'=hD^Hc_i$.
Consequently, $2nP|\hh_j|^2 \sim \chi^2(2,2nP|h|^2\frac{|\langle c_i, c_j\rangle|^2}{(nP)^2})$, $2mse_j \sim \chi^2(2n-2,2\delta)$, where
\begin{align}
\delta &= \sum_{k=1}^{n-1}|c'_k|^2 = ||\cc'||^2 - |c'_n|^2 \notag \\ 
&= |h|^2nP - |h|^2\frac{|\langle c_i,c_j\rangle|^2}{nP}.
\end{align}
Furthermore, $2nP|\hh_j|^2 $ and $mse_j$ are independent because $\{N'_k\}_{k=1}^{n-1}$ and $N'_n$ are independent.
\end{proof}

\begin{re}
\label{re1}
Let $X\sim \chi^2(n,\delta)$, we have $E(X)=n+\delta$, $var(X)=2n+4\delta$. By Chebyshev inequality, $X$ is concentrated at 
$\left[E(X)-t\sqrt{var(X)}, E(X)+t\sqrt{var(X)}\right]$. Therefore, let $ \DD_j=\{\yy: mse_j \leq T\}$ for some properly chosen threshold $T$, we can identify different receivers successfully provided that $\frac{|\langle c_i,c_j \rangle|^2}{(nP)^2}$ are not too close to $1$.  For simplicity of theoretical analysis, we use only $mse$ to construct decoding regions in asymptotic analysis, and $\hh$ will be utilized to boost finite-length performance.   
\end{re}

By exploring the maximum number of codewords with not-too-large pairwise inner products, we analyze the achievable rates for both discrete and continuous input symbols. The results are quite promising. Asymtotically, deterministic idenification capacity exceeds channel capacity even in the absence of CSI. 
\begin{theo}
\label{theo1}
Assume $\forall \ \epsilon > 0$, $\exists \ 0 < \uh < \oh$, such that $P\left(\uh \leq |h|^2 \leq \oh \right)>1-\epsilon$.\\
If inputs are transmitted by QPSK symbols, we have
\begin{equation}
\label{eq4}
\lim\limits_{n \rightarrow \infty}\frac{\log\SM}{n} = 2,
\end{equation} 
if inputs are transmitted by continuous complex symbols, we have
\begin{equation}
\label{eq5}
\liminf\limits_{n \rightarrow \infty}\frac{\log\SM}{n\log n} \geq \frac{1}{2}. 
\end{equation}
\end{theo}
\begin{proof}
Let $Z_1 \sim \chi^2(2n-2)$, $Z_2^h \sim \chi^2(2n-2,2|h|^2nP(1-\frac{|\langle c_i,c_j \rangle|^2}{(nP)^2})$ be the $mse$ distributions of target and non-target receivers. If $\exists \ 0 < b \leq 1, \ K >0$, such that
\begin{equation}
\label{eq3} 
1-\frac{|\langle c_i,c_j \rangle|^2}{(nP)^2} \geq \frac{K}{n^{\frac{1}{2}(1-b)}}, \ \forall \  1 \leq i\neq j \leq M.
\end{equation}
By Chebyshev inequality, if $\uh \leq |h|^2 \leq \oh$, we have
\begin{align}
&P[Z_1 > 2n+t_1\sqrt{4n}] \leq \frac{1}{t_1^2},  \label{eq2a}\\
&P[Z_2^h \leq 2n+t_1\sqrt{4n}] \leq P[Z_2^h \leq 2n-2+ \notag\\ 
&2\uh PKn^{\frac{1}{2}(1+b)}- t_2\sqrt{(4+8\oh P)n}]  \label{eq2} \leq \frac{1}{t_2^2},
\end{align} 
where (\ref{eq2}) is satisfied for sufficiently large $n$. To bound the type-I and type-II error rates below $\lambda_1$ and $\lambda_2$, let $\epsilon = \frac{\lambda_2}{2}$, $t_1 = \sqrt{\frac{1}{\lambda_1}}$, $t_2 = \sqrt{\frac{2}{\lambda_2}}$, $T = 2n + t_1\sqrt{4n}$. Construct decoding regions as $\DD_i = \{2mse \leq T\}$, we have
\begin{align}
\label{eq10}
P(\DD_i^c|c_i)& = P(Z_1 > T) \leq \lambda_1, \notag \\  
P(\DD_j|c_i) &\leq P(Z_2^h \leq T, \ \uh \leq |h|^2 \leq \oh) +  P(|h|^2 \notin [\uh,\oh])\notag\\
&  \leq \lambda_2.
\end{align}
According to (\ref{eq10}), $\{c_i\}_{i=1}^M$ form an $\left(n,M,P,\lambda_1, \lambda_2\right)$ code if (\ref{eq3}) is satisfied for some $K$ and $b$. Next, we analyze the maximum number of codewords satisfying (\ref{eq3}) for some $K$ and $b$ specified in the following proofs. 

To prove (\ref{eq4}), we denote the ball $\mathcal{A}(c,\rho)=\{ c': |\langle c,c' \rangle| > nP\sqrt{1-\rho} \}$, $0 < \rho < 1$ as a set of sequences $c'$ that are less orthogonal to $c$. If $c_i \notin \mathcal{A}(c_j,\rho)$, $\forall i \neq j$, then (\ref{eq3}) is satisfied with $K=\rho$, $b = 1$.

 Firstly, we prove $|\mathcal{A}(c,\rho)|$ are the same for different $c$. Let \\$c = (\sqrt{P}e^{j(\frac{\pi}{4}+\frac{\pi}{2}c_1)},\cdots,\sqrt{P}e^{j(\frac{\pi}{4}+\frac{\pi}{2}c_n)})^T$ and $a = (\sqrt{P}e^{j(\frac{\pi}{4}+\frac{\pi}{2}a_1)},\cdots,\sqrt{P}e^{j(\frac{\pi}{4}+\frac{\pi}{2}a_n)})^T$ be arbitrary QPSK sequences.  
Let $c' = (\sqrt{P}e^{j(\frac{\pi}{4}+\frac{\pi}{2}c'_1)},\cdots,\sqrt{P}e^{j(\frac{\pi}{4}+\frac{\pi}{2}c'_n)})^T$, $a' = (\sqrt{P}e^{j(\frac{\pi}{4}+\frac{\pi}{2}a'_1)},\cdots,\sqrt{P}e^{j(\frac{\pi}{4}+\frac{\pi}{2}a'_n)})^T$, where $a'_i=a_i+c'_i-c_i$, we have 
\begin{align}
\langle c,c' \rangle = \sum_{i=1}^n Pe^{j\frac{\pi}{2}(c_i-c'_i)} = \langle a,a' \rangle,
\end{align}
thus $c' \in \mathcal{A}(c,\rho)$ if and only if $a' \in \mathcal{A}(a,\rho)$. In other words, we construct a bijection between $\mathcal{A}(c,\rho)$ and $\mathcal{A}(a,\rho)$, therefore, $|\mathcal{A}(c,\rho)|= |\mathcal{A}(a,\rho)|$. 

Denote by $N(\rho) = |\mathcal{A}(c,\rho)|$ the number of QPSK sequences in the ball $\mathcal{A}(c,\rho)$, and w.l.o.g., assume $c = (e^{j\frac{\pi}{4}},\cdots,e^{j\frac{\pi}{4}})^T$, we are going to provide an upper bound on $N(\rho)$. If $n_{c'}^0 \geq n\epsilon(\rho)$ and $n_{c'}^1 \geq n\epsilon(\rho)$, where $\er = \frac{1-\sqrt{1-\rho}}{2-\sqrt{2}}$, we have
\begin{align}
\label{eq6}
|\langle c,c' \rangle| &\leq  \sqrt{2}nP\er + nP(1-2\er) \notag \\
&= nP(1-(2-\sqrt{2})\er) = nP\sqrt{1-\rho}.
\end{align}
Similarly, (\ref{eq6}) is satisfied if any two QPSK symbol numbers exceed $n\er$. Therefore, $c' \in \mathcal{A}(c,\rho)$ only if there is a QPSK symbol in overwhelming number, i.e.,
\begin{align}
\mathcal{A}(c,\rho) \subset \{c': \exists \ k, n_{c'}^k > n(1-3\er)\}.
\end{align}
Thus
\begin{align}
\label{eqnc}
N_c(\rho) &\leq |\{c': \exists \ k, n_{c'}^k > n(1-3\er)\}| \notag \\
& \leq 4\times C_n^{n(1-3\er)} 4^{3n\er} \\
&\leq 4 \times 2^{nH(3\er)}\times 4^{3n\er} = 2^{n\left(6\er+H(3\er)\right)+2}. \notag
\end{align}
Following the proof of Gilbert-Varshamov (GV) bound \cite{b14}, if $\exists$ $c_{k+1}\notin \bigcup_{i=1}^k \mathcal{A}(c_i,\rho)$, we can add $c_{k+1}$ to $\{c_i\}_{i=1}^k$ to form larger codes satisfying (\ref{eq3}). This packing process continues until the balls $\mathcal{A}(c_k,\rho)$ cover all the QPSK sequences. Thus $\SM \geq \frac{4^n}{N_c(\rho)}$,
\begin{align}
\frac{\log\SM}{n} \geq 2-6\er-H(3\er)-\frac{2}{n}.
\end{align}
Notice $\er \rightarrow 0$ when $\rho \rightarrow 0$, thus
\begin{align}
\liminf\limits_{n \rightarrow \infty} \frac{\log\SM}{n} \geq 2,
\end{align}
$\limsup\limits_{n \rightarrow \infty} \frac{\log\SM}{n} \leq 2$ is satisfied trivially, thus (\ref{eq4}) is proved.

To prove (\ref{eq5}), firstly, we transform the problem of analyzing the maximum number of codewords satisfying (\ref{eq3}) into analyzing the maximum number of non-overlapping ball packing. By Cauchy-Schwarz inequality, $\frac{|\langle c_i,c_j \rangle|^2}{(nP)^2}=1$ if and only if $c_j=e^{j\theta}c_i$ for some $\theta \in [0,2\pi)$. Therefore, (\ref{eq3}) is satisfied if the Euclidean distance between $c_j$ and $e^{j\theta}c_i$ is large enough $\forall$ $\theta \in [0,2\pi)$. Therefore, if
\begin{align}
||c_j-e^{j\theta}c_i|| \geq  \sqrt{2Pn^{\frac{b+1}{2}}} \triangleq r_n , \forall \  \theta \in [0,2\pi),
\end{align}
let $\theta=-\arg\langle c_i,c_j \rangle$, we have
\begin{align}
\frac{|\langle c_i,c_j\rangle|^2}{(nP)^2} &= \left(\frac{2nP-||c_j-e^{-j\arg\langle c_i,c_j\rangle }c_i||^2}{2nP}\right)^2 \notag \\
& \leq (1-\frac{1}{n^{\frac{1-b}{2}}})^2 \leq 1-\frac{1}{n^{\frac{1-b}{2}}},
\end{align}
so (\ref{eq3}) is satisfied with $K=1$. Denote 
\begin{align}
\mathcal{U}(c,r_n)= \{c' \in \MS_0(n,\sqrt{nP}): \bigcup_{\theta}\BB_{e^{j\theta}c}(n,r_n)\}, 
\end{align}
where $\mathcal{U}(c,r_n)$ is the intersection of the spherical surface $\MS_0(n,\sqrt{nP})$ and the union of the spheres $\bigcup\limits_{\theta}\BB_{e^{j\theta}c}(n,r_n)$.  
Let $M$ be the maximum number of non-overlapping $\mathcal{U}(c,r_n)$, and $\{\mathcal{U}(c_i,r_n)\}_{i=1}^M$ is a saturate packing arrangement, we have $\{c_i\}_{i=1}^M$ satisfy (\ref{eq3}), thus 
\begin{align}
\SM \geq M.
\end{align} 
Following the proof in \cite[Theorem 1]{b12}, we prove (\ref{eq5}) through analyzing the maximum number of non-overlapping ball packing $\{\mathcal{U}(c_i,r_n)\}_{i=1}^M$. If $c' \in \MS_0(n,\sqrt{nP})$ satisfies
\begin{align}
||c'-e^{j\theta}c_i|| \geq 2r_n, \ \forall \theta \in [0,2\pi)
\end{align}
 then $\{\mathcal{U}(c_i,r_n)\}_{i=1}^M \bigcup \mathcal{U}(c',r_n)$ is a larger non-overlapping arrangement, which is a contradiction. Thus 
\begin{align}
\MS_0(n,\sqrt{nP}) \subset \bigcup_i \mathcal{U}(c_i,2r_n),
\end{align}
and 
\begin{align}
M\times S[\mathcal{U}(c,2r_n)] \geq S(n,\sqrt{nP}).
\end{align}
For sufficiently large $n$, we have 
\begin{align}
S[\mathcal{U}(c,2r_n)] \leq 2\pi CnPV(n-1,2r_n),
\end{align}
where $C$ is a constant, so
\begin{align}
\SM &\geq \frac{S(n,\sqrt{nP})}{2\pi CnPV(n-1,2r_n)}\notag \\
&=\frac{1}{C}(\frac{nP}{4r_n^2})^{n-1}= \frac{1}{C}(\frac{n^{\frac{1-b}{2}}}{8})^{n-1}.
\end{align}
Thus 
\begin{align}
\liminf\limits_{n \rightarrow \infty}\frac{\log\SM}{n\log n} \geq \frac{1-b}{2}, \forall \ 0 < b \leq 1,
\end{align}
so (\ref{eq5}) is proved.
\end{proof}

\begin{re}
The proof for (\ref{eq4}) can be generalized to higher order modulation with finite constellation number.
Because an $n$-dimensional complex sphere is equivalent to a $2n$-dimensional real sphere, the achievable rate in (\ref{eq5}) is $\frac{1}{2}$ compared to  $\frac{1}{4}$ in \cite{b12}. The above theorem reveals that achievable identification rate \emph{without} CSI is equal to that \emph{with} CSI. This result is reasonable due to the following. As the blocklength goes to infinity, channel estimation error becomes negligible, so the identification capacity loss due to inaccurate channel estimation becomes negligible too.
\end{re}

\subsection{Finite Blocklength Analysis}
\label{sub1}
In practice, fading coefficient $h$ varies within a long blocklength, thus the finite-length analysis is critical for practical applications. For identification, the error analysis is based on characterizing the inner product, which is different from Polyanskiy's finite-length analysis for transmission codes \cite{b16}. Actually, the analysis for deterministic identification is simpler.

Missed detection usually causes a much severer damage to a communication system than a false activation. In this regard, we choose to bound the type-I error rate below $\lambda_1$ and minimize type-II error rate. According to (\ref{eq72}) and Theorem \ref{theo1}, we can identify receivers using the $mse$ metric, the decoding region of receiver-$j$ with the $mse$ metric is defined as 
$$\DD_j^{mse} = \{ mse_j \leq F^{-1}_{2n-2}(1-\lambda_1)/2\},$$
thus the type-I error rate is equal to $\lambda_1$ regardless of $h$.

In this section, we focus on analyzing the type-II error rates of QPSK identification codes. Let $N_d = |\{ c': \frac{|\langle c,c' \rangle|^2}{P^2} \geq d\}|$ be the number of QPSK sequences satisfying $\frac{|\langle c,c' \rangle|^2}{P^2} \geq d$. W.l.o.g., we assume $c = (e^{j\frac{\pi}{4}},\cdots,e^{j\frac{\pi}{4}})^T$, thus
\begin{align}
N_d &=  |\{c': \left(n_{c'}^0-n_{c'}^2\right)^2+\left(n_{c'}^1-n_{c'}^3\right)^2 \geq d)\}| \notag \\
&= \sum_{\left(n_{c'}^0-n_{c'}^2\right)^2+\left(n_{c'}^1-n_{c'}^3\right)^2 \geq d} \frac{n!}{n_{c'}^0!n_{c'}^1!n_{c'}^2!n_{c'}^3!}.
\end{align}
Denote by $\lambda_2(d)$ the type-II error rate between $c'$ and $c$ when $\frac{|\langle c, c'\rangle|^2}{P^2}=d$. Following the proof of GV bound in \cite[Theorem 5.1.7]{b14}, we obtain a type-II error rate upper bound for optimal  QPSK identification codes. 

\begin{theo}
Let $d_{min} = \min\{d: N_{d+1} < \frac{4^n}{M-1}\}$, then $(n,M,P,\lambda_1,\lambda_2^{GV}= \lambda_2(d_{min}))$ QPSK identification codes exist.
\end{theo}
\begin{proof}
We construct an $(n,M,P,\lambda_1,\lambda_2^{GV})$ identification code $\CC_M$ by adding new QPSK sequences sequentially. Firstly, let $\CC_1 = \{c_1\}$, $c_1$ can be any QPSK sequence. Assume we have $\CC_k = \{c_1,\cdots,c_{k}\}$, $1 \leq k \leq M-1$ satisfying $\frac{|\langle c_i,c_j \rangle|^2}{P^2} \leq d_{min}$, $\forall 1 \leq i\neq j \leq k$. Because $kN_{d_{min}+1}<4^n$, there is a QPSK sequence $c_{k+1}$ with $\frac{|\langle c_{k+1},c_j \rangle|^2}{P^2} \leq d_{min}$ to all codewords of $\CC_k$. Let $\CC_{k+1} = \{c_1,\cdots,c_{k+1}\}$, thus $\frac{|\langle c_i,c_j \rangle|^2}{P^2} \leq d_{min}$, $\forall 1 \leq i\neq j \leq k+1$. Finally, we obtain $\CC_M = \{c_1,\cdots,c_M\}$ satisfying $\frac{|\langle c_i,c_j \rangle|^2}{P^2} \leq d_{min}$, $\forall 1 \leq i\neq j \leq M$ and $\CC_M$ is an $(n,M,P,\lambda_1,\lambda_2^{GV})$ QPSK identification code.  
\end{proof}
Furthermore, we use the average weight spectrum of random codes to approximate $\ola_2$ as $\lambda^{approx}_2 = \sum_{d=0}^{n^2} P_d \lambda_2(d)$, where $P_d = \frac{N_d}{4^n}$ is the probability that random QPSK sequences $c$, $c'$ satisfy  $\frac{|\langle c, c'\rangle|^2}{P^2}=d$. Simulation results of random QPSK identification codes in Rayleigh fading channels are in Fig. \ref{fig.3}. As shown in Fig. \ref{fig.3}, $\lambda_2^{GV}$ and $\lambda^{approx}_2$ are accurate type-II error rate estimations for random QPSK identification codes in finite-length regime. The average type-II error rate $\lambda_2^{ave}$ does not increase with the number of users, i.e., the average type-II error rate is almost constant as the coding rate increases but cannot be arbitrary small at finite-length. Therefore, the average error rate is mainly affected by the blocklength rather than the coding rate. It is consistent with the asymptotic analysis \cite{b4}, where the coding rate of the identification codes can be infinite, and the achievable type-II error rate tends to zero as the code length increases.      

\begin{figure}[htbp]
\centerline{\includegraphics[width = .5\textwidth]{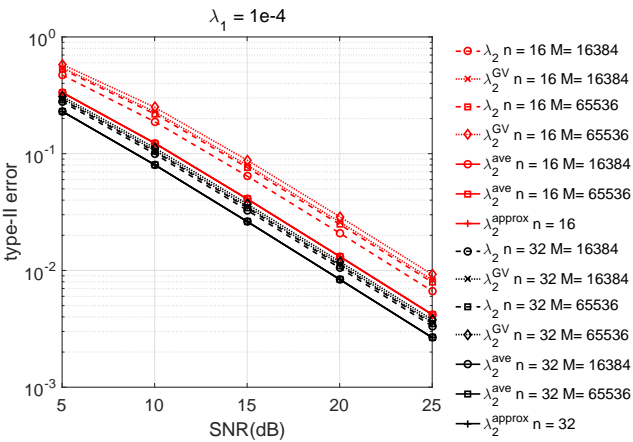}}
\caption{Type-II error rates of random codes using the $mse$ metric.}
\label{fig.3}
\end{figure}

\subsection{Blind Identification in Pilot-Assisted Transmission}
In the following application example, we consider blind identification codewords $\{c_i\}_{i=1}^{M}$ applied as pilot symbols in PAT schemes. The transmitter encodes a message into a codeword $\xx_{m}$ by the channel code $\mathcal{C}(N,K)$, and the overall sequence  $\xx=(c_i,\xx_{m})$ is transmitted over block fading channel. Denote by $\yy = (\yy_{id},\yy_{m})$ the received sequence, each active receiver will use $\yy_{id}$ to estimate $h$ and  identify if it is the target receiver. A type-I error occurs when the target receiver discards a message that could have been decoded correctly, so
\begin{align}
p_1 \leq \int_0^{\infty} P(\yy \in \DD_i^c||h|=x)(1-\text{bler}(x))p(x)dx,
\end{align} 
where $p_1$ is the overall type-I error rate, $p(x)$ is the distribution density of $|h|$, and $\text{bler}(x)$ is the block error rate of $\mathcal{C}(N, K)$ at $\text{SNR}=Px^2$ with perfect channel estimation. 

According to (\ref{eq72}) and (\ref{eq82}), if $|h|$ is small, the $mse$ distributions between target and non-target receivers are indistinguishable, thus the blind identification performance would be inferior with small $|h|$. However, the overall type-I error rate is small due to the high BLER. We can further reduce type-II error rate by allowing receivers to use both $|\hh_i|$ and $mse_i$, and we call it the $two\text{-}look$ metric. The decoding region of receiver-$j$ with the $two\text{-}look$ metric is defined as  
$$\DD_j^{two\text{-}look} = \{ mse_j \leq T \ \text{and} \ |\hh_j| \geq \oh \}$$ 
for properly chosen $T$ and $\oh$ specified in section \ref{simu}. By Lemma \ref{le1}, the type-II error rate of UE-$j$ with the \emph{two-look} metric is 
\begin{align*}
p_2 = \int_0^{\infty} P(mse_j \leq T||h|=x)P(|\hh_j| \geq \oh ||h|=x)p(x)dx.
\end{align*}
The key point in the construction of $\DD_i$ is that $\forall$ $h$, $P(|\hh_j| \geq \oh)$ and $P(mse_j \leq T)$ can not be large at the same time, thus the type-II error rate will be small regardless of $h$. Constructing  $\DD_i$ only on $P(|\hh_j| \geq \oh)$ $(P(mse_j \leq T))$ will result in high type-II error rate when $|h|$ is large (small).

\section{Simulation Results}
\label{simu}
In this section, we provide simulation results to verify the efficiency of the proposed blind deterministic identification method in PAT. Assume $h \sim \mathcal{CN}(0,1)$, so $|h| $ follows Rayleigh distribution. We use (128, 64) polar codes\cite{b15} as transmission codes, with CRC-Aided (CA) successive cancellation list (SCL) decoding \cite{b18}\cite{b19}. Identification codes are randomly generated as in Section \ref{sub1}. Let $T=F^{-1}_{2n-2}(1-\frac{\lambda_1}{2})/2$, $p_1$ is bounded by
\begin{align}
\label{eq9}
p_1 &\leq p(\uh)(1-\text{MC}(\uh)) + \int_{\uh}^{\infty}P(|\hh| \leq \oh||h|=x)p(x)dx \notag \\
& + P(mse > T),
\end{align}
where MC is the meta converse bound \cite{b17}, $\oh$ is chosen to bound (\ref{eq9}) below $\lambda_1$, and $\uh$ is optimized to maximize $\oh$. Denote by $\lambda_1^{true}$ the overall type-I error rate, we have $\lambda_1^{true} \leq \lambda_1$ according to (\ref{eq9}), and the value of $\lambda_1^{true}$ can be obtained from simulation results in Fig. \ref{fig.5}. As shown in Fig. \ref{fig.4}, the increase of miss detection probability is negligible. And the average typer-II error rate with the $two$-$look$ metric is two orders of magnitude lower than the type-II error rate with the $mse$ metric as shown in Fig. \ref{fig.7}.  

\begin{figure}[htbp]\centering
\subfigure[BLER $+$ $\lambda_1^{true}$]{
\label{fig.4}
\includegraphics[width = .22\textwidth]{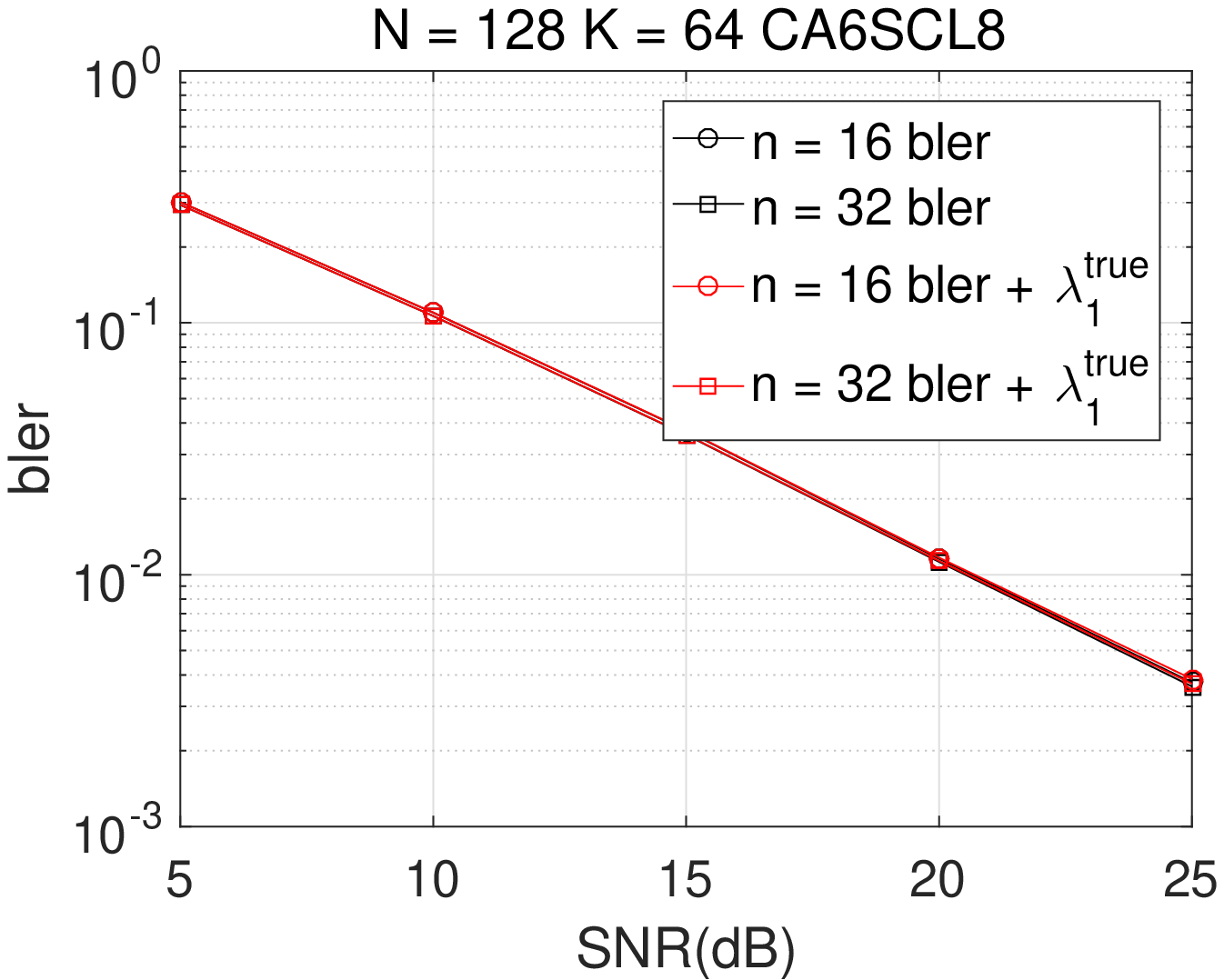}
}
\hfill
\subfigure[$\lambda_1^{true}$]{
\label{fig.5}
\includegraphics[width = .22\textwidth]{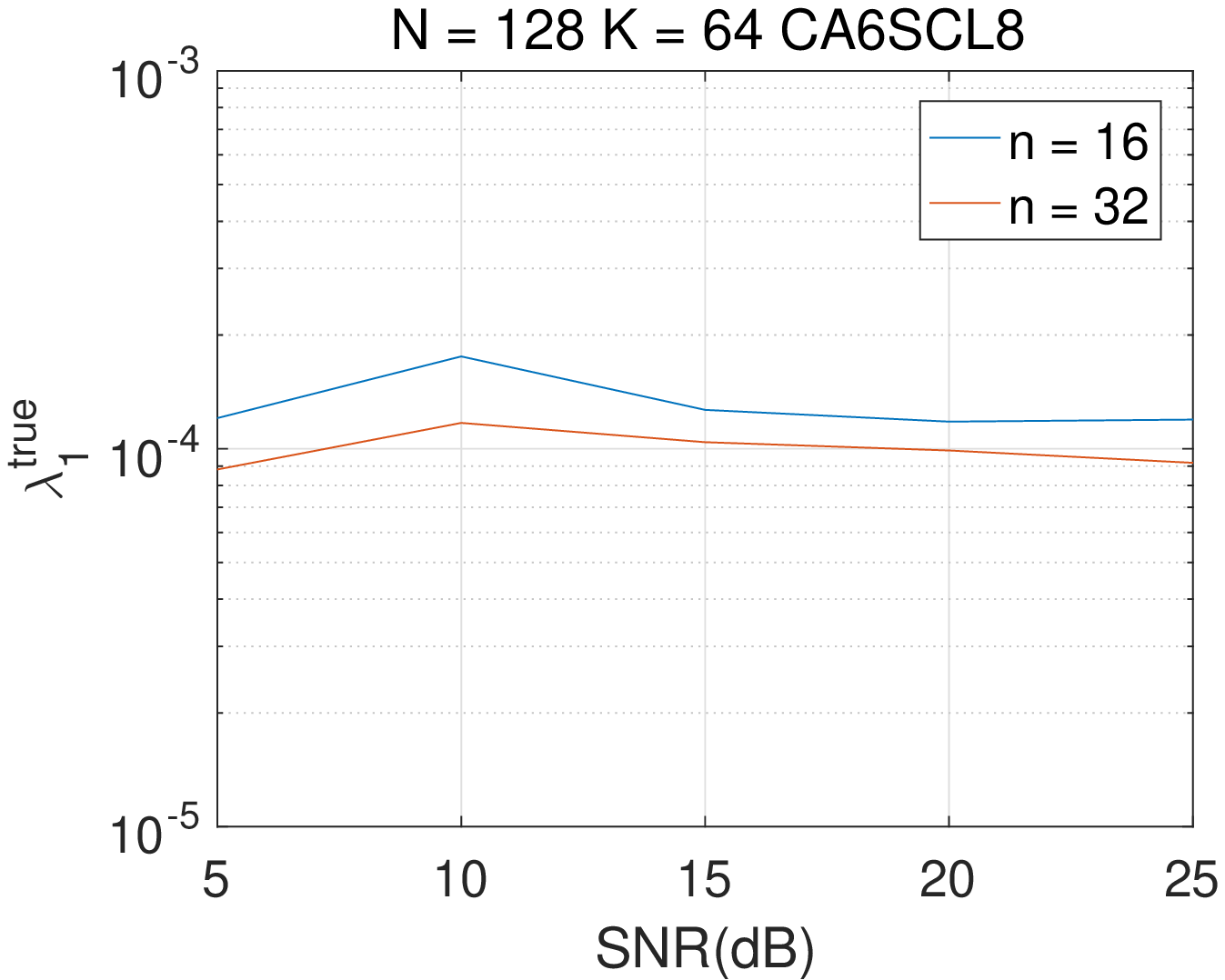}
}
\subfigure[Type-II error rates of random codes]{
\label{fig.7}
\includegraphics[width = .5\textwidth]{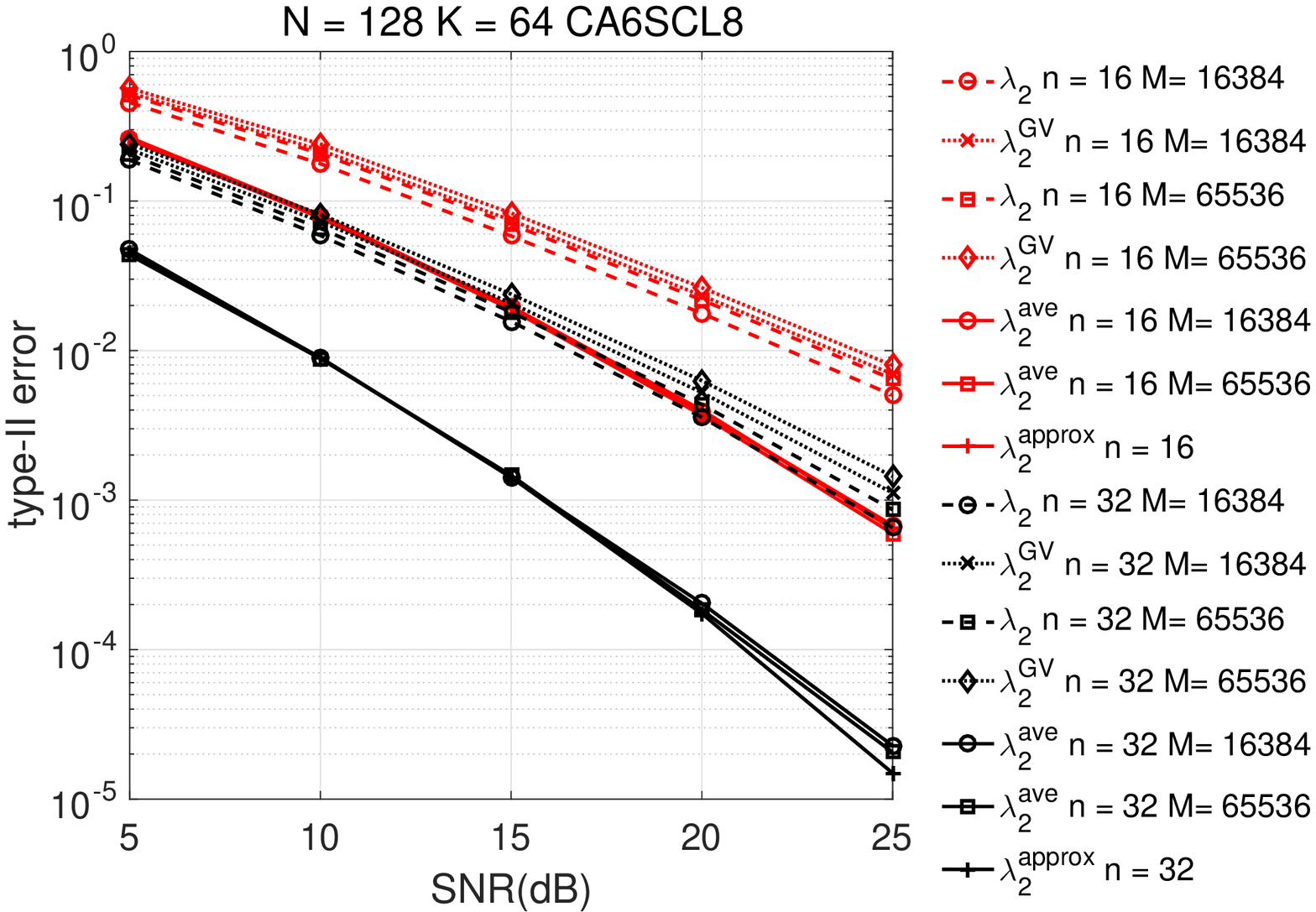}
}
\caption{Identification performance of random codes using the $two$-$look$ metric.}
\label{fig.6}
\end{figure}

\section{Conclusion}
In the paper, we generalize deterministic identification to the practical scenarios by considering block fading channels \emph{without} CSI and finite-length performance. We prove that the identification capacity \emph{without} CSI can still be higher than transmission capacity asymptotically. And the finite-length performance is also analyzed. In practice, we use the blind identification codes for both channel estimation and user identification in PAT schemes. The identification metric is optimized to bound the type-I error rate while minimize the type-II error rate. Simulation results verify the efficiency of the deterministic identification codes in finite blocklength regime.

\vspace{12pt}

\end{document}